%% file: main.tex
\RequirePackage[l2tabu,orthodox]{nag}
\documentclass
[11pt,letterpaper]
{article}

\newcommand{\Erdos}{Erd\H{o}s\xspace}
\newcommand{\Renyi}{R\'enyi\xspace}

\usepackage[notes=true,later=false,camera=false]{dtrt}
\usepackage[utf8]{inputenc}
\usepackage{etex}
\usepackage{ stmaryrd }
\usepackage{xspace,enumerate}
\usepackage[T1]{fontenc}
\usepackage[full]{textcomp}
\usepackage[american]{babel}
\usepackage{mathtools}

\usepackage{amsthm}
\usepackage{empheq}

   \usepackage{hyperref}
\usepackage[capitalise,nameinlink]{cleveref}
\crefname{lemma}{Lemma}{Lemmas}
\crefname{fact}{Fact}{Facts}
\newcommand{\colorconstraints}{\text{Color Constraints}}
\crefname{colorconstraints}{(color constraints)}{Color Constraints}
\crefformat{colorconstraints}{#2\colorconstraints#3}
\crefname{indsetconstraints}{(indset constraints)}{IndSet Constraints}
\crefformat{indsetconstraints}{#2$\mathsf{IndSet\ Axioms}$#3}
\crefname{theorem}{Theorem}{Theorems}
\crefname{mtheorem}{Theorem}{Theorems}
\crefname{corollary}{Corollary}{Corollaries}
\crefname{claim}{Claim}{Claims}
\crefname{fact}{Fact}{Fact}
\crefname{example}{Example}{Examples}
\crefname{algorithm}{Algorithm}{Algorithms}
\crefname{problem}{Problem}{Problems}
\crefname{definition}{Definition}{Definitions}
\usepackage{paralist}
\usepackage{turnstile}
\usepackage{mdframed}
\usepackage{tikz}
\usepackage{caption}
\DeclareCaptionType{Algorithm}
\usepackage{newfloat}
\newtheorem{theorem}{Theorem}[section]
\newtheorem*{theorem*}{Theorem}

\newtheorem*{proposition*}{Proposition}
\newtheorem{lemma}[theorem]{Lemma}
\newtheorem*{lemma*}{Lemma}

\newtheorem*{conjecture*}{Conjecture}
\newtheorem{fact}[theorem]{Fact}
\newtheorem*{fact*}{Fact}

\newtheorem*{hypothesis*}{Hypothesis}

\theoremstyle{definition}
\newtheorem{definition}[theorem]{Definition}
\newtheorem*{definition*}{Definition}

\theoremstyle{remark}

\newtheorem*{claim*}{Claim}

\newtheorem*{remark*}{Remark}

\newtheorem*{observation*}{Observation}

\usepackage[
letterpaper,
top=1.2in,
bottom=1.2in,
left=1in,
right=1in]{geometry}
\usepackage{newpxtext} %
\usepackage{textcomp} %
\usepackage[varg,bigdelims]{newpxmath}
\usepackage[scr=rsfso]{mathalfa}%
\usepackage{bm} %
\let\mathbb\varmathbb
\usepackage{microtype}

\allowdisplaybreaks

\newcommand{\bignorm}[1]{\big\lVert#1\big\rVert}
\newcommand{\Bignorm}[1]{\Big\lVert#1\Big\rVert}

\newcommand{\R}{{\mathbb R}}
\newcommand{\N}{{\mathbb N}}
\newcommand{\norm}[1]{\lVert #1 \rVert}

\let\epsilon=\varepsilon

\newcommand{\E}{{\mathbb E}}
\newcommand{\1}{\mathbf{1}}

\newcommand{\C}{\mathbb C}

\newcommand{\Var}{\mathrm{Var}}

\newcommand{\cB}{\mathcal B}

\newcommand{\ceil}[1]{\lceil #1 \rceil}

\newcommand{\mper}{\,.}
\newcommand{\mcom}{\,,}
\newcommand{\Paren}[1]{\left(#1\right)}

\newcommand{\Brac}[1]{\left[#1\right]}

\newcommand{\Norm}[1]{\left\lVert#1\right\rVert}

\newcommand{\iprod}[1]{\langle#1\rangle}
\newcommand{\inprod}[1]{\langle#1\rangle}

\newcommand{\cD}{\mathcal D}

\newcommand{\cut}{\mathsf{cut}}

\newcommand{\SRPC}{\mathsf{SRPC}}

\newcommand{\cin}{\mathrm{in}}
\newcommand{\cout}{\mathrm{out}}

\newcommand{\Oh}{O}

\begin{document}

\newcommand{\FormatAuthor}[3]{
\begin{tabular}{c}
#1 \\ {\small\texttt{#2}} \\ {\small #3}
\end{tabular}
}
\input{title}

\maketitle
\thispagestyle{empty}

\begin{abstract}
  \input{abstract}

\end{abstract}

\clearpage

\pagestyle{plain}
\setcounter{page}{1}
\input{introduction}

\input{phalf}

\input{pgeneral}
\input{discussion}

\section*{Acknowledgements}
JB, RB, and DS received funding from the European Research Council (ERC) under the European Union’s Horizon 2020 research and innovation programme (grant agreement No 815464). PK was supported by NSF CAREER Award \#2047933, NSF \#2211971, an Alfred P. Sloan Fellowship, and a Google Research Scholar Award.

\bibliographystyle{alpha}
\bibliography{bib/custom,bib/dblp,bib/mathreview,bib/scholar,bib/references,bib/witmer}
\appendix
\input{appendix}

\end{document}

%% file: title.tex
\title{Semirandom Planted Clique\\and the Restricted Isometry Property}

\author{
\begin{tabular}[h!]{cc}
  \FormatAuthor{Jaros\l aw B\l asiok}{jaroslaw.blasiok@inf.ethz.ch}{ETH Zürich}&
  \FormatAuthor{Rares-Darius Buhai}{rares.buhai@inf.ethz.ch}{ETH Zürich}\\\\
  \FormatAuthor{Pravesh K.\ Kothari}{kothari@cs.princeton.edu}{IAS \& Princeton University}&
  \FormatAuthor{David Steurer}{dsteurer@inf.ethz.ch}{ETH Zürich}
\end{tabular}
} %
\date{\today}

%% file: abstract.tex
We give a simple, greedy $O(n^{\omega+0.5})=O(n^{2.872})$-time algorithm to list-decode planted cliques in a semirandom model introduced in~\cite{DBLP:conf/stoc/CharikarSV17} (following~\cite{FK01}) that succeeds whenever the size of the planted clique is $k\geq O(\sqrt{n} \log^2 n)$. In the model, the edges touching the vertices in the planted $k$-clique are drawn independently with probability $p=1/2$ while the edges not touching the planted clique are chosen by an adversary in response to the random choices. Our result shows that the computational threshold in the semirandom setting is within a $O(\log^2 n)$ factor of the information-theoretic one~\cite{steinhardt2017does} thus resolving an open question of Steinhardt. This threshold also essentially matches the conjectured computational threshold for the well-studied special case of fully random planted clique.

All previous algorithms~\cite{DBLP:conf/stoc/CharikarSV17,MMT20,MR4617517-Buhai23} in this model are based on rather sophisticated rounding algorithms for entropy-constrained semidefinite programming relaxations and their sum-of-squares strengthenings and the best known guarantee is a $n^{O(1/\epsilon)}$-time algorithm to list-decode planted cliques of size $k \geq \tilde{O}(n^{1/2+\epsilon})$. In particular, the guarantee trivializes to quasi-polynomial time if the planted clique is of size $O(\sqrt{n} \operatorname{polylog} n)$. Our algorithm achieves an almost optimal guarantee with a surprisingly simple greedy algorithm.

The prior state-of-the-art algorithmic result above is based on a reduction to certifying bounds on the size of unbalanced bicliques in random graphs --- closely related to certifying the restricted isometry property (RIP) of certain random matrices and known to be hard in the low-degree polynomial model. Our key idea is a new approach that relies on the truth of --- but not efficient certificates for --- RIP of a new class of matrices built from the input graphs.

%% file: introduction.tex
\section{Introduction}%
\label{sec:introduction}

Finding planted cliques~\cite{DBLP:journals/rsa/Jerrum92,DBLP:journals/dam/Kucera95}  in \Erdos-\Renyi random graphs $G(n,1/2)$ is a well-studied average-case variant of the notoriously hard clique problem. The added clique is uniquely identifiable with high probability (and recoverable via a brute-force quasi-polynomial time algorithm) whenever it has size $k \gg 2 \log n$. However, the best known efficient algorithms require $k \gg \sqrt{n}$. The algorithms themselves are rather simple. The simple greedy degree heuristic of reporting the $k$ largest-degree vertices works if $k \gg \sqrt{n \log n}$ and the $\sqrt{\log n}$ factor can be shaved off by a natural spectral algorithm~\cite{MR1662795-Alon98}. Lower bounds in restricted models~\cite{BHK+16,MR3664576-Feldman17} provide evidence that efficient algorithms cannot beat the $\sqrt{n}$ threshold and this \emph{information-computation} gap is the fountainhead of several hardness results in average-case optimization~\cite{brennan2018reducibility, brennan2020reducibility}.

The simple algorithms above for finding planted cliques are rather brittle --- modifying only $O(k^2)$ edges that do not even touch the planted clique is enough to completely break their guarantees. To disallow such brittle heuristics that tend to ``overfit'' to the specific idealized random model, Feige and Kilian~\cite{FK01} introduced \emph{semirandom} models~\cite{bwca-semirandom} for the planted clique problem (following the seminal work of Blum and Spencer~\cite{BLUM1995204}). The input graph in their model is chosen by a combination of benign random and adaptive adversarial choices that preclude brittle heuristics while still hopefully steering clear of the worst-case hard instances. Such a semirandom model is the main focus of this work:

\begin{definition}[Semirandom planted clique, $\SRPC(n,k,p)$~\cite{DBLP:conf/stoc/CharikarSV17}]
\label{def:src}
To an empty graph $G([n],E)$:
\begin{enumerate} 
\item \emph{Plant a clique}: Plant a clique on an arbitrary subset of vertices $S^* \subseteq [n]$ with $|S^*|=k$. 
\item \emph{Include cut edges at random}: Add each edge in $\cut(S^*)$ independently with probability $p$.
\item \emph{Choose rest of the edges adversarially}: Adaptively choose any induced graph on $[n] \setminus S^*$.
\end{enumerate}
\end{definition} 
In the fully random planted clique setting, the edges in the third step are also chosen at random. Note further that the adversarial choice is \emph{adaptive}, i.e., made in response to the random choice in the second step. Approximating the maximum clique in such a graph is clearly as hard as the worst-case variant~\cite{MR1687331-Haastad99,MR2403018-Zuckerman07}, since in the third step we can simply plant a worst-case hard instance on $[n] \setminus S^*$. The ``right'' goal, instead, is to find a $k$-clique in the input graph. In fact, all known algorithms focus on the formally stronger goal of finding a small ($\sim n/k$ in the strongest possible results) \emph{list} of $k$-cliques guaranteed to contain $S^*$. This \emph{list-decoding} goal relaxes the usual \emph{unique recovery} goal that is clearly impossible in the semirandom model.

A new course of progress on the problem was begun in~\cite{DBLP:conf/stoc/CharikarSV17} who gave a polynomial-time algorithm based on rounding an appropriate semidefinite programming relaxation that works whenever $k \geq O(n^{2/3} \log^{1/3} n)$. McKenzie, Mehta and Trevisan~\cite{MMT20} improved this threshold by logarithmic factors to $k \geq O(n^{2/3})$. Recently~\cite{MR4617517-Buhai23} made further progress by finding an $n^{O(1/\epsilon)}$-time algorithm that succeeds whenever the planted clique has size $k \geq n^{1/2+\epsilon}$ for any $\epsilon>0$. The results in the latter two works also tolerate a monotone adversary that can delete an arbitrary subset of edges in $\cut(S^*)$. 
On the flip side, Steinhardt~\cite{steinhardt2017does} proved that if $k = o(\sqrt{n})$ it is \emph{information-theoretically} impossible to recover a list of size $O(n/k)$ that is guaranteed to contain the planted clique with high probability.\footnote{Observe that a disjoint union of $k$-cliques with all the remaining edges chosen independently with probability $1/2$ is a valid input instance in the model, with $n/k$ indistinguishable $k$-cliques. Thus, the minimum list size possible is $\geq n/k$.}\footnote{The lower bound instance formally adds a clique that is slightly larger than $k$ --- hence, the input $k$ is a \emph{lower bound} on the size of the planted clique. All algorithmic results including ours easily generalize to only need a lower bound on size of the planted clique.}
That is, the information-theoretic threshold for the natural estimation task is \emph{shifted} from $\lceil 2 \log_2 n \rceil$ to $\sim \sqrt{n}$ in the semirandom model. If this threshold were also achievable by a polynomial-time algorithm, Steinhardt noted, then the widely conjectured information-computation gap in the fully random planted clique problem disappears for the natural \emph{robust} version of the problem studied in this (and his) work. With this motivation, Steinhardt~\cite{steinhardt2017does} (and for a closely related variant, Feige~\cite{bwca-semirandom}) explicitly poses the natural question of whether semirandom planted clique admits an efficient algorithm at $k \sim \tilde{O}(\sqrt{n})$.

When $k = \sqrt{n} \operatorname{polylog} n$, the state-of-the-art algorithm above devolves into a near-brute-force running time of $n^{O(\log n / \log \log n)}$. Further, even for $k \geq n^{1/2+\epsilon}$, all the known results above~\cite{DBLP:conf/stoc/CharikarSV17,MMT20,MR4617517-Buhai23} are based on rather sophisticated rounding algorithms for semidefinite programming relaxations (and their sum-of-squares strengthenings) with \emph{entropy maximization} constraints. This raises the natural question of whether simple heuristics could succeed at $k = \tilde{O}(\sqrt{n})$ or at least match the best known results. 

\paragraph{This work.} In this work, we give a surprisingly simple greedy algorithm that finds a list of size $(1+o(1)) n/k$ containing the planted clique with high probability, whenever $k \geq O(\sqrt{n} \log^2 n)$, essentially resolving Steinhardt's question. Our algorithm admits an optimized implementation via black-box usage of fast matrix multiplication and runs in time $O(n^{\omega+0.5}) = O(n^{2.872})$.

\begin{theorem}[Main result] \label{thm:main-intro}
There is an $O(n^{\omega+0.5})$-time algorithm, where $\omega \leq 2.372$ is the matrix multiplication exponent, that takes as input a graph $G \sim \SRPC(n,k,1/2)$ and, if $k \geq O(\sqrt{n} \log^2 n)$, outputs a list of $k$-cliques of size $(1+o(1)) n/k$ such that, with probability at least $0.99$ over the draw of $G$ and the random choices of the algorithm, $S^*$ is contained in the list. 
\end{theorem}
Observe that our result comes within a $O(\log^2 n)$ factor of the information-theoretic lower bound~\cite{steinhardt2017does}.\footnote{Steinhardt suggests that the lower bound on the threshold can possibly be improved to $k \geq O(\sqrt{n \log n})$ --- the threshold at which the natural quasi-polynomial time brute-force algorithm works. This would make the threshold of our algorithm off by a $\log^{1.5} n$ factor.}
\cref{thm:pgeneral-main} presents an appropriate generalization to $\SRPC(n,k,p)$ for arbitrary $p$ that improves upon the bound on $k$ in previous results when $p \leq 1-n^{-0.01}$.

\paragraph{A new connection to RIP.} Our algorithm relies on a \emph{new} connection to the \emph{restricted isometry property}\footnote{An $m \times n$ matrix $H$ is said to satisfy $(r,\delta)$-RIP if for every $r$-sparse vector $v$, $\Norm{Hv}_2 \in (1 \pm \delta) \Norm{v}_2$.} of a certain $n^{O(1)}$-size matrix built from the graph that relies only on the randomness of the edges in the cut defined by the planted clique. We show how to relate the success of a simple greedy algorithm to the RIP property of this matrix that, crucially, \emph{does not need} any efficient certificates for the RIP itself. This is crucial because we need the RIP property in the near optimal (and conjecturally, hard to efficiently certify) parameter regimes. 

To put this discussion in proper context, recall that the fully-random planted clique problem is closely related to efficiently certifying sparse quadratic forms (i.e., certifying upper bounds on $v^{\top} Av$ for the adjacency matrix $A$ of the input graph and sparse $v$) or sparse PCA but these reductions also produce instances in the conjectured hard regimes~\cite{berthet2013computational,HKP+17,DBLP:journals/tit/KoiranZ14}. Indeed, such results are usually interpreted as hardness results for certifying the RIP and bounds on sparse quadratic forms (assuming the planted clique conjecture). To use such reductions algorithmically, we need efficient certificates of RIP in conjectured hard parameter regimes~\cite{DBLP:journals/tit/BandeiraDMS13,MR3245368-Koiran14,MR4345126-Ding21}\footnote{For example, while an $m \times n$ matrix with independent $\pm 1 / \sqrt{m}$ entries is known to be $(r,0.1)$-RIP for $r \sim \Omega(m/\log (n/m))$, efficiently certifying such a fact is conjectured to be hard~\cite{DBLP:journals/tit/BandeiraDMS13,MR3245368-Koiran14,MR4345126-Ding21}.}.

As an important special case,  we note that the problem of certifying bounds on the biclique numbers of random graphs that forms the crux of the approach for semirandom planted clique in~\cite{MR4617517-Buhai23} also reduces to certifying RIP in the conjectured hard regime.

Our reduction, in contrast, shows how to solve the semirandom planted clique problem as long as some matrices (rather different from the ones arising in the above reductions) built from the input graph satisfy RIP. Crucially, we \emph{do not need} efficient certificates of RIP --- this important difference allows our approach to obtain a significantly better guarantee in this work. 

\paragraph{Monotone deletions.} As we noted above, while the original work of Charikar, Steinhardt and Valiant~\cite{DBLP:conf/stoc/CharikarSV17}, the open question of Steinhardt~\cite{steinhardt2017does}, and the lower bound~\cite{steinhardt2017does} are all phrased about the model studied in this paper, the results in~\cite{MMT20,MR4617517-Buhai23} tolerate, in addition, an adversary that can delete an arbitrary subset of edges in $\cut(S^*)$~\cite{FK01}. Our approach needs certificates of RIP (in the conjectured hard regime) to tolerate such a monotone adversary. A similar bottleneck  prevents algorithms based on certifying biclique numbers to reach the $k=\sqrt{n} \operatorname{polylog} n$ threshold. This is made formal via low-degree polynomial lower bounds in~\cite{MR4617517-Buhai23}.

Our result raises an intriguing possibility of a shift in the \emph{computational threshold} for a natural estimation problem due to a monotone adversary. This is in contrast to the intriguing work of Moitra, Perry and Wein~\cite{DBLP:conf/stoc/MoitraPW16} who showed that the \emph{information-theoretic threshold} for community detection in the stochastic block model \emph{does shift} under a monotone adversary.

\section{Overview of our algorithm}
In this section we give a brief overview of our approach for semirandom planted clique. 
As in the prior work~\cite{MR4617517-Buhai23} (see~\Cref{lem:pruning} in the Appendix), when $k \geq O(\sqrt{n \log n})$ we can prune any polynomial-size list containing $S^*$ down to size $(1 + o(1))n/k$ without removing $S^*$. This can be achieved by removing from the list all cliques that have intersection larger than $\Omega(\log n)$ with any other clique in the list.  %
Therefore we will focus on how to recover a list of $k$-cliques containing $S^*$ of polynomial size.  

\paragraph{Recalling the approach of~\cite{MR4617517-Buhai23}.} Before describing our approach, let us recall the key idea in the prior state-of-the-art~\cite{MR4617517-Buhai23} that gave a polynomial-time algorithm that works whenever $k \geq n^{1/2+\epsilon}$ for an arbitrary constant $\epsilon >0$.
Their main idea is a reduction (within the sum-of-squares framework) to efficiently ``certifying'' bounds on the biclique numbers of bipartite random graphs. 
In particular, in order for their algorithm to succeed at $k = \tilde{O}(\sqrt{n})$, they need a constant-degree sum-of-squares certificate (roughly, polynomial-time verifiable via an SDP  relaxation) that a $k \times n$ bipartite random graph (each edge included independently with probability $1/2$) has no $O(\log^{C} n) \times k$ bipartite clique in it for $C = O(1)$ (with high probability). Crucially, they need certificates even for the model \emph{without monotone deletions}. This certification task is naturally related (and, in fact, reduces) to \emph{certifying} the restricted isometry property of matrices built from the bipartite random graph in the conjectured-to-be-hard regime~\cite{DBLP:journals/tit/BandeiraDMS13,MR3245368-Koiran14}. Unfortunately, they also show that certifying such biclique numbers likely suffers from information-computation gaps by establishing lower bounds in the low-degree polynomial model, and thus appear to run into an inherent barrier in obtaining algorithms that work when $k = \sqrt{n} \log^{O(1)}(n)$.

Our main idea is a new approach that circumvents the need for certificates on the biclique number of bipartite random graphs altogether and succeeds with just the truth of (but without efficient certificates for) the restricted isometry property. 
Our approach leads to an algorithm that does not need semidefinite programming and is based on a simple, efficient, greedy procedure, the correctness of which relies on the RIP of a natural matrix built from the input graph.
In the following, we explain how to build up to this procedure by starting with a naive greedy procedure. 

\subsection{The naive greedy algorithm} 
We let $G$ denote the $\pm 1$-adjacency matrix of the graph $G$, i.e., $G(i,j)=1$ iff $\{i,j\}$ is an edge in $G$. 
We also let $G_i$ be the $i$-th row of $G$, $G_i^{\cin} \in \{\pm 1\}^k$ be the projection of $G_i$ to coordinates in $S^*$, and $G_i^{\cout} \in \{\pm 1\}^{n-k}$ the projection of $G_i$ to coordinates in $[n] \setminus S^*$. 

Let us make the simple observation that $G_i$ and $G_j$ are non-trivially correlated if $i$ and $j$ both are in $S^*$. Indeed, we have $\iprod{G_i,G_j} = \iprod{G_i^{\cin}, G_j^{\cin}} + \iprod{G_i^{\cout},G_j^{\cout}}$. Since $i,j \in S^*$, the first term is clearly $k$.
Further, since every edge in $\cut(S^*)$ is chosen independently to be in $G$ with probability $1/2$, $G_i^{\cout}$ and $G_j^{\cout}$ are uniformly random and independent elements of $\{\pm 1\}^{n-k}$ and thus, with $1-1/n^{O(1)}$ probability over the draw of edges in $\cut(S^*)$, $|\iprod{G_i^{\cout}, G_j^{\cout}}| \leq O(\sqrt{n \log n})$.
Thus, $\iprod{G_i,G_j} \geq k \pm O(\sqrt{n \log n}) \geq k/2$ if $k \geq O(\sqrt{n\log n})$. 

This observation naturally suggests the following simple algorithm:

\vspace{4mm}
\noindent \textbf{Naive greedy procedure:} For a uniformly random $i \in [n]$, add $S_i = \{j: \iprod{G_i,G_j} \geq k/2\}$ to the list.

\vspace{4mm}
We choose an $i \in S^*$ with probability $k/n$ and, thus, repeating the above procedure $O(n/k)$ times ensures that with probability at least $0.99$ we pick some $i \in S^*$.
From the above correlation computation, $S^* \subseteq S_i$ for any $i \in S^*$.
Could such an $S_i$ contain a $j \not \in S^*$?

For an $i \in S^*$ and $j \not \in S^*$, we can write $\iprod{G_i,G_j} = \iprod{G_i^{\cin},G_j^{\cin}} + \iprod{G_i^{\cout},G_j^{\cout}}$.
The first term is at most $O(\sqrt{k \log n})$ with high probability.
However, we have little control on the second term since $G_j^{\cout}$ is chosen by an adversary (in response to the random choice of edges in $\cut(S^*)$).
In fact, it turns out that the adversary can arrange (multiple!) $j \not \in S^*$ such that $\iprod{G_i,G_j} \gg k$ \emph{simultaneously} for $\geq O(n^2/k^2)$ different $i \in S^*$ --- see~\Cref{lem:adv-cor}.
Notice that $n^2/k^2 \gg k$ if $k = \tilde{O}(\sqrt{n})$, in which case the adversary can ensure that \emph{every} $i\in S^*$ fails to produce $S_i = S^*$ in the above greedy algorithm.

\paragraph{Analyzing naive greedy algorithm for $k \geq \tilde{O}(n^{3/4})$.}
While the naive greedy procedure fails for $k = \tilde{O}(\sqrt{n})$, 
we now argue that it does succeed for $k \geq \tilde{O}(n^{3/4})$. This will form the starting point of our new approach. 

The following simple lemma uses standard bounds on random matrices to control the number of spurious $j \not \in S^*$ that can be in $S_i$ for $i \in S^*$. 
 
\begin{lemma} \label{lem:overview-main-analysis}
Let $v_1, v_2, \ldots, v_k \in \{ \pm 1\}^n$ be uniformly random and independent vectors for $k \geq O(\sqrt{n \log n})$.
Then with probability at least $1-1/n$ over the draw of vectors $v_i$, for every $u \in \{ \pm 1 \}^n$ there are at most $O(n^2/k^2)$ vectors $v_i$ such that $\iprod{u,v_i} \geq k/3$. 
\end{lemma}
\begin{proof}[Proof of \cref{lem:overview-main-analysis}]
Consider the $n \times k$ matrix $H$ with columns $v_1, v_2, \ldots, v_k$.
Then, by standard results in random matrix theory, $\Norm{H}_2 \leq O(\sqrt{n} + \sqrt{k})$ with probability $1-1/n$.
Let $\cB \subseteq [k]$ be such that, for every $i \in \cB$, $\iprod{u,v_i} \geq k/3$. 
Then, by the Cauchy-Schwartz inequality, we have
\[
\iprod{u, H\1_{\cB}} \leq \|u\|_2 \|H \1_{\cB}\|_2 \leq \sqrt{n} \sqrt{|\cB|} \|H\|_2 \leq O(n \sqrt{\cB})\,.
\]
On the other hand, by the choice of set $\cB$, we have
\[
\iprod{u, H\1_{\cB}} = \sum_{i \in \cB} \iprod{u, v_i}  \geq |\cB| k / 3\,.
\]
Combining those two inequalities and rearranging, we get $|\cB| \leq O(n^2/k^2)$.  
\end{proof}
Using this lemma, it is easy to show that if $k \gg n^{3/4}$ then for a uniformly random $i \in S^*$ we have $|S_i \setminus S^*| \leq o(k)$:
We apply \cref{lem:overview-main-analysis} with vectors $G_i^{\cout}$ for $i \in S^*$ --- each a uniformly random and independent element of $\{ \pm 1\}^{n-k}$.
Then, taking $u = G_j^{\cout}$ for any $j \in [n]$, we know that there are at most $O(n^2/k^2)$ different $i \in S^*$ such that $\iprod{G_i^{\cout}, G_j^{\cout}} \geq k/3$.
Thus, there are in total $O(n^3/k^2)$ pairs $(i,j)$ such that $i \in S^*$ and $j \not \in S^*$ with $\iprod{G_i,G_j} \geq k/2 \geq k/3 + O(\sqrt{k \log n})$.
In particular, by averaging, a uniformly random $i \in S^*$ satisfies $\iprod{G_i,G_j} \geq k/2$ for at most $C n^3/k^3$ different $j\not \in S^*$ (for some fixed constant $C>0$) with probability at least $0.99$.
If $k \gg O(n^3/k^3)$ (e.g., if $k \geq \omega(n^{3/4})$), then $|S_i \setminus S^*| \leq o(k)$.
Finally, from any such set $S_i \supseteq S^*$ that has at most $o(|S^*|)$ erroneous elements, one can obtain $S^*$ by a simple pruning procedure --- for instance by keeping only vertices that have at least $k-1$ neighbors in $S_i$ (see, e.g., Claim 5.9 on page 35 of~\cite{MR4617517-Buhai23}). 

\subsection{Tensoring twice: greedy matches the guarantees of~\cite{MR4617517-Buhai23}}

The simple greedy procedure above 1) relies on standard spectral norm bounds on random matrices with independent entries and 2) gets stuck at $k \geq \tilde{O}(n^{3/4})$.
We now show how to obtain an algorithm that works for any $k \geq n^{1/2+ \epsilon}$ in time $n^{\Oh(1/\varepsilon)}$, matching the guarantees of~\cite{MR4617517-Buhai23} via a simple greedy algorithm analogous to the one in previous subsection. 

For $\alpha \subset [n]$ of size $|\alpha|=2$, let $G_{\alpha} \in \{\pm 1\}^n$ be the vector so that $G_{\alpha}(i) = G_{\alpha_1,i} \cdot G_{\alpha_2,i}$.
That is, $G_{\alpha}(i)$ is the product of the $\pm 1$-indicators of the two edges $(\alpha_1, i)$ and $(\alpha_2, i)$.
Note that $G_{\alpha}$ is the Hadamard (and not the tensor) product of $G_{\alpha_1}$ and $G_{\alpha_2}$. 

Similarly to the previous section, it is easy to observe that for every $\alpha, \alpha' \subset S^*$, $\iprod{G_{\alpha}, G_{\alpha'}} \geq k/2$ if $k \geq O(\sqrt{n \log n})$.
This motivates the following generalization of the greedy procedure above:

\vspace{4mm}

\noindent \textbf{``Tensored'' greedy procedure:} For a uniformly random $\alpha \in {{[n]} \choose 2}$, add $S_\alpha = \{ i : \iprod{G_i, G_{\alpha}} \geq k/2\}$ to the list.

\vspace{4mm}
Let us analyze this algorithm.
The same argument as before shows that $S_{\alpha} \supseteq S^*$ if $\alpha \subset S^*$.
Further, a uniformly random $\alpha$ is in $S^*$ with probability $(k/n)^2$, so repeating the above procedure $O(n/k)^2$ times includes a uniformly random sample of $\alpha \subset S^*$. 

We will now provide a generalization of \cref{lem:overview-main-analysis} with no change in parameters that allows reasoning about $|S_{\alpha} \setminus S^*|$.
The proof, however, will rely on the RIP instead of the straightforward spectral norm argument above. 

\begin{lemma} \label{lem:overview-main-analysis-2-tensor}
Let $v_1, v_2, \ldots, v_k \in \{ \pm 1\}^n$ be uniformly random and independent vectors for $k \geq O(\sqrt{n \log n})$ and let $v_{\alpha} \in \{ \pm 1\}^n$ for $\alpha \subset [k]$ of size $|\alpha|=2$ be defined by $v_{\alpha}(i) = v_{\alpha_1} (i) v_{\alpha_2}(i)$ for every $i$.
Then with probability at least $1-1/n$ over the draw of vectors $v_i$, for every $u \in \{ \pm 1 \}^n$ there are at most $O(n^2/k^2)$ vectors $v_\alpha$ such that $\iprod{u,v_\alpha} \geq k/3$. 
\end{lemma}

\paragraph{Tensored greedy succeeds for $k \geq \tilde{O}(n^{0.6})$.}
Before proving \cref{lem:overview-main-analysis-2-tensor}, let us repeat the calculations we did above to see if we improve the range of $k$ where our algorithm succeeds.
By applying the lemma above to the vectors $G_\alpha^{\cout}$ for $\alpha \in S^*$, we obtain that the number of pairs $(\alpha,j)$ such that $\iprod{G_{\alpha}^{\cout}, G_j^{\cout}} \geq k/3$ and $j \not \in S^*$ is at most $O(n^3/k^2)$.
Thus, by averaging, at least $1/2$ of the $\sim k^2$ different sets $\alpha$ must satisfy $\iprod{G_{\alpha}^{\cout},G_j^{\cout}} \geq k/3$ for at most $O(n^3/k^4)$ different $j$.
In particular, $|S_\alpha \setminus S^*|\leq O(n^3/k^4)$ for $1/2$ fraction of $\alpha \subset S^*$.
This bound is $\ll k$ if $n^3/k^4 \ll k$ or $k \gg n^{3/5} = n^{0.6}$. 

Observe that this already improves on the $n^{3/4}$ bound above.
In fact, this simple greedy algorithm already improves on the (more involved) algorithms in prior works~\cite{DBLP:conf/stoc/CharikarSV17,MMT20} that work when $k \gg n^{2/3} (\gg n^{3/5}$)! 

\paragraph{Matching the guarantee of ~\cite{MR4617517-Buhai23}.}
We now show a simple modification of the greedy procedure above that works to give an $n^{O(1/\epsilon)}$-size list when $k \sim n^{1/2+\epsilon}$, allowing us to match the running-time vs $k$ trade-off of the recent work~\cite{MR4617517-Buhai23}. 

\vspace{4mm}
\noindent \textbf{Procedure:} Fix $t \in \N$. For independent and uniformly random $\alpha_1, \alpha_2, \ldots, \alpha_t \in \binom{[n]}{[2]}$, add 
$S_{\alpha_1, \alpha_2, \ldots, \alpha_t} = \{j : \iprod{G_j,G_{\alpha_i}}\geq k/2 \text{ for every } 1 \leq i \leq t\}$
to the list.  
\vspace{3mm}

Note that tensored greedy procedure from before corresponds to $t=1$ in the above algorithm. 
Observe that, with probability $(k/n)^{2t}$, $\alpha_1, \alpha_2, \ldots, \alpha_t \subset S^*$, so repeating the procedure $O(n/k)^{2t}$ times includes a random subset of $S^*$.
As before, for such tuples $S_{\alpha_1, \alpha_2, \ldots, \alpha_t} \supseteq S^*$.
We will show that for a $0.99$-fraction of tuples $\alpha_1, \alpha_2, \ldots, \alpha_t\subset S^*$, we have $S_{\alpha_1, \alpha_2, \ldots, \alpha_t} = S^*$ if $k \geq O(n^{1/2+ 1/2t})$.
Taking $t = 1/\epsilon$ yields the trade-off obtained in~\cite{MR4617517-Buhai23}.

Let us see why: from \cref{lem:overview-main-analysis-2-tensor}, for every $j \not \in S^*$, the number of sets $\alpha \subset S^*$ such that $\iprod{G_\alpha, G_j} \geq k/2$ is at most $O(n^2/k^2)$.
Thus, if we pick $t$ uniformly random subsets $\alpha_i \subset S^*$, the chance that for a given $j$ we have $\iprod{G_{\alpha_i},G_j} \geq k/2$ for every $1 \leq i \leq t$ is at most $((n^2/k^2)/k^2)^{t} = (n^2/k^4)^t$.
If $k \geq 100 n^{1/2 + 1/2t}$, then, this chance is clearly $\ll 1/n$ and thus, by a union bound over the choice $j \not \in S^*$, with probability at least $0.99$ over the choice of the $\alpha_i$ we have $S_{\alpha_1, \alpha_2, \ldots, \alpha_t} = S^*$.

\paragraph{Proving \cref{lem:overview-main-analysis-2-tensor} using RIP.}
Let us introduce the star of the show to prove \cref{lem:overview-main-analysis-2-tensor}: the restricted isometry property. 
\begin{definition}[Restricted isometry property] \label{def:RIP}
An $m \times d$ matrix $H$ is said to be $(r,\delta)$-RIP if for every $v \in \R^d$ such that $\Norm{v}_0 \leq r$ we have
\begin{equation*}
(1-\delta) \|v\|_2 \leq \|H v\|_2 \leq (1+\delta) \|v\|_2\,.
\end{equation*}
\end{definition}

In fact, for the purpose of the analysis of the algorithm, we care only about the upper bound in the restricted isometry property. More concretely:
\begin{definition}[Sparse operator norm] \label{def:weak-RIP}
An $m \times d$ matrix $H$ is said to be $r$-sparse operator norm bounded by $C$ if for every $v \in \R^d$ such that $\Norm{v}_2\leq 1$, $\Norm{v}_0 \leq r$, we have $\Norm{H v}_2 \leq C$. 
\end{definition}
We note that sparse operator norm behaves differently compared to the sparse quadratic form in our setting (i.e., maximizing $v^{\top}Hv$ over sparse vectors for square $H$). The latter occurs naturally in the context of sparse PCA and has connections to the planted clique problem~\cite{DBLP:conf/colt/BerthetR13}.

Clearly, a bound on the sparse operator norm is implied by the RIP --- if a matrix $H$ is $(r,\delta)$-RIP, then its $r$-sparse operator norm is bounded by $1+\delta$.
Now, it is well-know that an $m \times d$ ($m \ll d$) matrix of independent, uniform $\pm 1$ entries satisfies $(r,C)$-RIP for $C =O(1)$ and $r = \Omega(m/\log (d/m))$.
We will actually establish and use the RIP of a random matrix with significantly correlated columns. 

\begin{fact}[See \Cref{lem:tensoring-gives-RIP}] \label{fact:weak-RIP-overview}
Let $H \in \R^{(n-k) \times {k \choose 2}}$ be a matrix with columns $G_{\alpha}^{\cout}$ for $\alpha \subset S^*$ of size $|\alpha|=2$. Then with high probability $H/\sqrt{n-k}$ satisfies $(r,\Oh(1))$-RIP for some $r \geq \Omega(n/\log^{O(1)}(n))$. 
\end{fact}

Observe that $H$ is a function of at most $kn$ random bits even as it has $nk^2$ entries. %
Despite this highly dependent setting, it turns out that $H$ still satisfies RIP. 
This follows from a powerful result of~\cite{MR2417886-Rudelson08} that allows us to conclude that a matrix with independent rows chosen from an isotropic, $\ell_{\infty}$-bounded distribution satisfies strong RIP properties. The following theorem is usually stated in terms of \emph{bounded orthonormal systems}. In \Cref{sec:rip}, we discuss how this standard formulation is equivalent to the following version that is directly interpretable and useful for us. 
\begin{theorem}[Corollary of \cite{MR2417886-Rudelson08, FR13}]
\label{thm:rip}
Let $A \in \mathbb{R}^{m \times N}$ be a random matrix with rows sampled i.i.d. according to the distribution of some random variable $\mathbf{X} \in \mathbb{R}^N$ that satisfies $\mathbb{E} \mathbf{X}\mathbf{X}^\top = I_N$ and $\norm{\mathbf{X}}_\infty \leq K$ almost surely.
Then, for any $m \geq O(r K^2 \log^3(r) \log(N) / \delta^2)$, we have with probability at least $1 - N^{-\log^3(r)}$ that the matrix $A/\sqrt{m}$ satisfies $(r, \delta)$-RIP. 
\end{theorem}

To finish this section, let us see how the sparse operator norm bound helps settle \cref{lem:overview-main-analysis-2-tensor}. 
\begin{proof}[Proof of \cref{lem:overview-main-analysis-2-tensor}]
Let $H$ be the matrix from \cref{fact:weak-RIP-overview}, with $r$-sparse operator norm of $H/\sqrt{n-k}$ bounded by $O(1)$ for $r \geq \Omega(n/\log^{O(1)}(n))$. 
Suppose (toward a contradiction) that there is a vector $u \in \{ \pm 1\}^{n-k}$ such that $\iprod{G_{\alpha},u} \geq k/3$ for all $\alpha \in \cB$, where $\cB$ is a set of size $r$ (if the set is larger than $r$ in size, just choose any $r$-size subset).
Then, using Cauchy-Schwarz and the sparse operator norm bound on $H$,
\[
\iprod{u/\Norm{u}_2, H \1_{\cB}/\sqrt{|\cB|}} \leq \Norm{H \1_{\cB}/\sqrt{|\cB|}}_2 \leq O(\sqrt{n-k})\,,
\]  
while by the choice of $\cB$ we have
\[
\iprod{u/ \|u\|_2, H \1_{\cB} / \sqrt{|\cB|}} \geq \frac{\sqrt{|\cB|} k}{3 \sqrt{n-k}}\,.
\]
Rearranging yields that $|\cB| \leq O(n^2/k^2)$. 
\end{proof}

\paragraph{Tensoring thrice: our algorithm.} At this point, it is intuitive to try the greedy procedure with tensoring thrice.
Indeed, this is precisely our algorithm, with the crux of the analysis being the RIP of a matrix analogous to the one above. Our proof is short and simple and is presented in full in the following section.

%% file: phalf.tex
\section{Semirandom planted clique with $p=1/2$}%
\label{sec:phalf}

We will prove the following result in this section. 
\begin{theorem}
\label{thm:phalf-main}
There exists an $O(n^{\omega+0.5})$-time algorithm, where $\omega \leq 2.372$ is the matrix multiplication exponent, that takes input a graph $G$ generated according to $\SRPC(n, k, 1/2)$ for any $k \geq O(\sqrt{n} \log^2 n)$ and with probability at least $0.99$ outputs a list of size $(1+o(1))n/k$ of $k$-cliques containing the planted clique in $G$.
\end{theorem}

We will first analyze the following simple algorithm:

\begin{mdframed}
\noindent \textbf{Algorithm 1:} \begin{enumerate} 
\item For $O((n/k)^3)$ rounds, sample $\alpha \subset [n]$ with $|\alpha|=3$ and construct $S_{\alpha} := \{j: \langle G_\alpha, G_j \rangle \geq k/2\}$.
\item Construct a refined set $S_{\alpha}' \subseteq S_{\alpha}$ by removing all vertices from $S_{\alpha}$ that are connected to less than $k - 1$ of the vertices in $S_{\alpha}$, and then add $S_{\alpha}'$ to the list if it forms a $k$-clique.
\item Apply pruning (see \cref{sec:pruning}) to reduce the list size to $(1+o(1))n/k$. \end{enumerate}
\end{mdframed}

This algorithm naively runs in time $O(n^{3.5})$. Later we will discuss how to implement computationally expensive steps in this algorithm by black-box calls to a matrix multiplication oracle, improving the running time to $O(n^{\omega + 0.5})$.

We will prove two claims: First, that for all $\alpha \subset S^*$ with $|\alpha|=3$ we have $S_{\alpha} \supseteq S^*$, and second, that for any fixed $j \not\in S^*$, there are at most $O(n^2/k^2)$ sets $\alpha \subset S^*$ with $|\alpha| = 3$ such that $j \in S_{\alpha}$. The first claim:

\begin{lemma}
\label{lem:all-good}
Let $k \geq O(\sqrt{n \log n})$. With high probability, for all $\alpha \subset S^*$ with $|\alpha|=3$ and all $j \in S^*$ with $j \not\in \alpha$, we have that $\langle G_\alpha, G_j\rangle \geq k / 2$.
\end{lemma}
\begin{proof}
We have that $\langle G_\alpha^{\cin}, G_j^\cin\rangle = k$.
On the other hand, $G_\alpha^{\cout}$ and $G_j^{\cout}$ have independent Rademacher entries, so by Hoeffding's inequality 
\[\mathbb{P}\Paren{|\langle G_\alpha^\cout, G_j^\cout \rangle| \geq t} \leq 2\exp\Paren{-\Omega(t^2/n)}\,.\]
Taking $t=O(\sqrt{n \log n})$ large enough, we have that $|\langle G_\alpha^\cout, G_j^\cout \rangle| \leq t$ with probability at least $1-1/n^{10}$.
We choose $k$ such that $t \leq k/2$, so $\langle G_\alpha, G_j\rangle \geq k-t \geq k/2$.
Finally, an application of the union bound shows that this holds simultaneously for all choices of $\alpha$ and $j$ with probability at least $1-o(1)$.
\end{proof}

Next, we want to prove that each $j \not\in S^*$ satisfies $\langle G_\alpha, G_j\rangle \geq k/2$ for at most $O(n^2/k^2)$ sets $\alpha \subset S^*$ with $|\alpha|=3$.
It is easy to bound $|\langle G_\alpha^\cin, G_j^\cin\rangle| \leq O(\sqrt{k \log n})$ using the fact that $G_j^\cin$ has independent Rademacher entries.
To talk about $\langle G_\alpha^\cout, G_j^\cout\rangle$, we define the matrix $H \in \{\pm 1\}^{(n-k)\times \Theta(k^3)}$ with columns $G_\alpha^\cout$ for all $\alpha \subset S^*$ with $|\alpha|=3$.
A simple observation shows that a bound on the $r$-sparse operator norm of the matrix $H / \sqrt{n-k}$ ensures the desired bound. 

\begin{lemma}
\label{clm:not-many-correlated-vectors}
Let $H \in \R^{q \times m}$ be a matrix satisfying
\begin{equation}
\sup_{\|v\|_2 \leq 1, \|v\|_0 \leq r} \|Hv\| \leq C\,, \label{eq:clm-rip}
\end{equation}
with $r \geq C^2 / \tau^2 + 1$. Then any vector $w \in \R^q$ with $\|w\|_2 \leq 1$ has inner product greater than or equal to $\tau$ with at most $C^2 / \tau^2$ columns of $H$.
\end{lemma}
\begin{proof}
Let $\cB \subset [m]$ be any set of size $|\cB| \leq r$, such that for all $t \in \cB$ we have $\inprod{H_{\cdot, t}, w} \geq \tau$. Then on one hand we have by Cauchy-Schwarz and~\eqref{eq:clm-rip} that
\begin{equation*}
    \inprod{H\mathbf{1}_\cB, w} \leq \|H \mathbf{1}_\cB\| \cdot \|w\| \leq  C\sqrt{|\cB|}\,,
    \end{equation*}
and on the other hand we have by the definition of $\cB$ that 
\begin{equation*}
    \inprod{H\mathbf{1}_\cB, w} = \sum_{t \in \cB} \inprod{H_{\cdot, t}, w} \geq \tau |\cB|\,.
    \end{equation*}
This leads, after rearranging, to $|\cB| \leq C^2/\tau^{2}$.
\end{proof}
If $H \in \R^{q \times m}$ was indeed a matrix with independent Rademacher entries, it is well-known that $H/\sqrt{q}$ would satisfy RIP of order $r = \Omega(q / \log(m/q))$.
For the matrix $H$ generated by taking the products of $3$-tuples of random signs that we care about, the columns are heavily dependent. It turns out, however, that we can still establish that $H / \sqrt{q}$ satisfies the restricted isometry property (with only a slight degradation in the sparsity parameter $r$).
\begin{lemma}
\label{lem:tensoring-gives-RIP}
Let $\mathcal{D}$ be any distribution over $\mathbb{R}$ such that $\mathbf{x} \sim \cD$ satisfies $\E \mathbf x = 0, \E \mathbf x^2 = 1$, and $|\mathbf x| \leq B$ almost surely.

Consider a matrix $H \in \mathbb{R}^{q \times \Theta(k^t)}$ with independent rows, where the columns are indexed by subsets $\alpha \subset [k]$ of size $t$ for constant $t$, and the rows are generated by first drawing $\mathbf{X}_i \sim \cD^k$ and setting $H_{i, \alpha} = \prod_{j \in \alpha} X_{i, j}$.
Then with high probability $H/\sqrt{q}$ satisfies $(r, \Oh(1))$-RIP with $r = \Omega(q / (B^{2t} \log^4(n)))$.
\end{lemma}
\begin{proof}
We want to apply~\Cref{thm:rip} to conclude that $H / \sqrt{q}$ satisfies RIP. To this end, all we need to show is that entries within any row of $H$ are uncorrelated and have variance one.

For any pair of distinct tuples $\alpha, \beta$, we have
\begin{align*}
\E [H_{i, \alpha} H_{i, \beta}]
&= \E \left[\Paren{\prod_{j \in \alpha} \mathbf X_{i,j}}\Paren{\prod_{j \in \beta} \mathbf X_{i,j}}\right]\\
&= \prod_{j \in \alpha\Delta \beta} \E[\mathbf X_{i,j}]\\
&= 0\,,
\end{align*}
where $\alpha \Delta \beta$ is the symmetric difference between sets $\alpha$ and $\beta$ (since for $j \in \alpha \cap \beta$ the term $\mathbf X_{i,j}^2$ appears in the monomial and $\E [\mathbf X_{i,j}^2] = 1$). 
Similarly $\E[H_{i,\alpha}^2] = \prod_{j\in \alpha} \E [\mathbf X_{i,j}^2] = 1$. Finally, since $|\mathbf X_{i,j}| \leq B$ for all $j$, we have $|H_{i, \alpha}| = \prod_{j \in \alpha} |\mathbf X_{i,j}| \leq B^t$.
Therefore the conditions of~\Cref{thm:rip} are satisfied for the matrix $H$, and the conclusion coincides with the lemma statement.
\end{proof}

We prove now the second claim.
\begin{lemma}
\label{lem:few-bad}
Let $k \geq O(\sqrt{n} \log^2 n)$. With high probability, for any $j \not\in S^*$, there exist at most $O(n^2/k^2)$ sets $\alpha \subset S^*$ with $|\alpha|=3$ such that $\langle G_\alpha, G_j\rangle \geq k / 2$.
\end{lemma}
\begin{proof}
Let $H \in \{\pm 1\}^{(n-k)\times \Theta(k^3)}$ be the matrix with columns $G_\alpha^\cout$ for all $\alpha \subset S^*$ with $|\alpha|=3$.
By~\Cref{lem:tensoring-gives-RIP}, the matrix $H/\sqrt{n-k}$ satisfies $(r,O(1))$-RIP for $r = \Omega((n-k)/\log^4(n))$.
Then we apply \Cref{clm:not-many-correlated-vectors} for $H/\sqrt{n-k}$ and $\tau = k/(3(n-k))$. Then no unit vector has correlation at least $\tau = k/(3(n-k))$ with more than $O(1 / \tau^2) \leq O(n^2/k^2)$ columns of $H/\sqrt{n-k}$.
Then also no vector in $\{\pm 1\}^{n-k}$ (of norm $\sqrt{n-k}$) has correlation at least $k/3$ with more than $O(n^2/k^2)$ columns of $H$.
We need to ensure that the conditions of \Cref{clm:not-many-correlated-vectors} are satisfied, namely that $r \gtrsim n^2/k^2$. This holds when $k \geq O(\sqrt{n} \log^2 n)$.

We can now finish the argument:
If a vector $G_j^\cout$ has correlation at most $k/3$ with $G_\alpha^\cout$, then also $\inprod{G_\alpha, G_j} \leq k/3 + O(\sqrt{k \log n}) < k/2$. 
Hence each vector $G_j$ with $j \not\in S^*$ has correlation greater than or equal to $k/2$ with at most $O(n^2/k^2)$ vectors $G_\alpha$.
\end{proof}
We are now ready to prove the correctness of Algorithm~1.
\begin{lemma}
\label{lem:phalf-correct}
Given a graph $G$ generated according to $\SRPC(n, k, 1/2)$ for any $k \geq O(\sqrt{n} \log^2 n)$, Algorithm~1 outputs with probability at least $0.99$ a list of size $(1+o(1))n/k$ of $k$-cliques containing the planted clique in $G$.
\end{lemma}
\begin{proof}

Consider a random $\alpha \subset S^*$ with $|\alpha|=3$.
\Cref{lem:all-good} implies that with high probability $S_\alpha \supseteq S^*$.
Furthermore, \Cref{lem:few-bad} says that, for each $j \not\in S^*$, with high probability $G_j$ has large inner product with at most $O(n^2/k^2)$ vectors $G_{\alpha}$ with $\alpha \subset S^*$ and $|\alpha|=3$.
Then $G_j$ has large inner product with $G_\alpha$ for random $\alpha \subset S^*$ with probability at most $O((n^2/k^2)/k^3) = O(1/k)$, using that there are $\Omega(k^3)$ choices for $\alpha$.
Hence the expected number of $G_j$ (for $j \not\in S^*$) that have large inner product with $G_\alpha$ is at most $O(n/k)$.
Therefore, by Markov's inequality, with probability at least $0.99$ over the choice $\alpha \subset S^*$ we have that $|S_{\alpha} \setminus S^*| \leq O(n/k) = o(k)$.

Once this is the case, note that removing vertices of degree smaller than $k-1$ in the induced graph given by $S_\alpha$ clearly does not affect vertices in $S^*$. On the other hand, for any $j \not \in S^*$, denoting by $N(j)$ the neighborhood of the vertex $j$ in $G$, we have
\begin{align*}
|N(j) \cap S_{\alpha}|
&\leq |N(j) \cap S^*| + |N(j) \cap (S_\alpha \setminus S^*)|\\
&\leq k/2 + \Oh(\sqrt{k \log n}) + o(k)\\
&< k-1\,.
\end{align*}
Then the pruning algorithm described in \Cref{sec:pruning} keeps $S^*$ in the final list, while reducing the size of the list to $(1 + o(1)) n/k$.
\end{proof}
We discuss now how, using fast matrix multiplication, Algorithm~1 can be implemented in time $O(n^{\omega + 0.5})$.
\begin{lemma}
\label{lem:phalf-fast}
Algorithm~1 can be equivalently implemented in time $O(T n^{\omega})$ where $T = \max((n/k)^3/n, 1)$ by invoking $O(T)$ calls to the $n\times n$ matrix-multiplication oracle. In particular, when $k \geq \sqrt{n}$, this yields $O(\sqrt{n})$ calls to the matrix multiplication oracle, and a total running time $O(n^{\omega + 0.5})$.
\end{lemma}
\begin{proof}
We can compute all inner products $\iprod{G_{\alpha}, G_j}$ simultaneously for all $O((n/k)^3)$ sampled $\alpha$: If we arrange the vectors $G_{\alpha}$ as the rows of a $O((n/k)^3) \times n$ matrix, and we multiply this matrix by the matrix $G \in \{\pm 1\}^{n\times n}$ --- the $\pm 1$-adjacency matrix of the graph --- then the $(\alpha,j)$-th entry of the product is exactly $\iprod{G_{\alpha}, G_j}$.
This matrix multiplication can be computed in time $(O((n/k)^3) / n) \cdot n^{\omega}$, by partitioning the set of all sampled triples $\alpha$ into $m = O((n/k)^3) / n$ parts of size $n$, and applying the $n\times n$ matrix multiplication oracle for each part of the partition. 
This allows us to compute the sets $S_{\alpha}$.

We can similarly use matrix multiplication to compute an $O((n/k)^3) \times n$ matrix $M$ where $M_{\alpha,j}$ represents the number of vertices in $S_{\alpha}$ that vertex $j$ is connected to --- this is by multiplying a $\{0,1\}$-valued matrix with rows being indicator vectors of sets $S_\alpha$, by the $\{0,1\}$-adjacency matrix of the graph $G$. Again, this can be done by using $O( (n/k)^3 / n)$ calls to the matrix multiplication oracle.
This allows us to also obtain the sets $S'_{\alpha}$.

Applying the exact same trick again, we can check which of the elements of $S'_\alpha$ are cliques --- we now construct a matrix $M'_{\alpha, j}$ denoting the number of neighbors vertex $j$ has in $S'_{\alpha}$ and keep only the sets $S'_\alpha$ such that for each $j \in S'_{\alpha}$ we have $M'_{\alpha, j} = |S'_{\alpha}| - 1$.

Therefore constructing the initial list involves $T = O((n/k)^3) / n \leq O(\sqrt{n})$ calls to the matrix multiplication oracle, and the size of the list is at most $O((n/k)^3)$.

Finally, by \Cref{lem:pruning-fast} the pruning algorithm can be implemented by $O((n/k)^3 / n)$ calls to the matrix multiplication oracle, leading to the overall time complexity bounded by $O(n^{\omega+0.5})$.
\end{proof}
The main theorem now follows immediately as a corollary of those two lemmas.
\begin{proof}[Proof of~\Cref{thm:phalf-main}]
The correctness of Algorithm~1 is given by~\Cref{lem:phalf-correct} and the time complexity $O(n^{\omega + 0.5})$ is given by~\Cref{lem:phalf-fast}.
\end{proof}

%% file: pgeneral.tex
\section{Semirandom planted clique with general $p$}%
\label{sec:pgeneral}
In this section we generalize our algorithm for $\SRPC(n,k,p)$ to arbitrary $1>p>0$.
Our bound on $k$ is better than in previous results~\cite{MMT20, MR4617517-Buhai23} when $p \leq 1 - n^{-0.01}$.
We note that the special case of $p=1/2$ already contains all the necessary ideas.

\begin{theorem}
\label{thm:pgeneral-main}
There exists an $O(n^{\omega+0.5})$-time algorithm, where $\omega \leq 2.372$ is the matrix multiplication exponent, that takes input a graph $G$ generated according to $\SRPC(n, k, p)$ for any\linebreak$k \geq O(\sqrt{n} \log^2 n \cdot \max(1, (p/(1-p))^4))$ and with probability at least $0.99$ outputs a list of size $(1+o(1))n/k$ of $k$-cliques containing the planted clique in $G$.
\end{theorem}

The first step is to consider a matrix $\bar{G}$ which is given by a construction analogous to $p$-biased characters: we want to shift and rescale the entries of the adjacency matrix so that in the random part the entries have mean zero and variance one. 
Specifically, let
\[\bar G(i, j) = \begin{cases}
    \sqrt{\frac{1-p}{p}} & \text{if } (i, j) \in E(G)\mcom\\
    -\sqrt{\frac{p}{1-p}} & \text{otherwise}\mper
  \end{cases}\]
An easy calculation confirms that indeed, if $\Pr((i, j) \in E(G)) = p$, then $\E[\bar{G}(i, j)]=0$ and ${\E[\bar{G}(i, j)^2] = 1}$. 
To simplify notation in the following argument we introduce $p_{+} = \sqrt{\frac{1 - p}{p}}$, $p_{-} = \sqrt{\frac{p}{1-p}}$, and $B = \max(p_{+}, p_{-})$ --- so that all entries of $\bar{G}$ are bounded in absolute value by $B$ (this property will turn out to be crucial later).

We also let $\bar G_i$ be the $i$-th row of $\bar G$, $\bar G_i^{\cin} \in \R^k$ be the projection of $G_i$ to coordinates in $S^*$, and $\bar G_i^{\cout} \in \R^{n-k}$ the projection of $\bar G_i$ to coordinates in $[n] \setminus S^*$. 
Finally, for a subset of rows $\alpha \subset [n]$, we write $\bar{G}_\alpha$ to denote the element-wise product of rows $\bar G_j$ for $j\in \alpha$.

We prove the same sequence of results as in the $p=1/2$ case:
First, that for all $\alpha \subset S^*$ with $|\alpha|=3$ we have $S_{\alpha} \supseteq S^*$, and second, that for any fixed $j \not\in S^*$, there are at most $O(n^2/k^2)$ sets $\alpha \subset S^*$ with $|\alpha| = 3$ such that $j \in S_{\alpha}$. The first claim:

\begin{lemma}
\label{lem:all-good-gen}
Let $k \geq O (\sqrt{n \log n} + B^4 \log n) / p_+^4$. 
With high probability, for all $\alpha \subset S^*$ with $|\alpha|=3$ and all $j \in S^*$ with $j \not\in \alpha$, we have that $\langle \bar G_\alpha, \bar G_j\rangle \geq kp_+^4/2$.
\end{lemma}
\begin{proof}
We have that $\langle \bar G_\alpha^{\cin}, \bar G_j^\cin\rangle = k p_+^4$. 
On the other hand, the element-wise product of $\bar G_\alpha^{\cout}$ and $\bar G_j^{\cout}$ has independent mean-zero and variance-one entries with maximum absolute value $B^4$, so by Bernstein's inequality 
\[\mathbb{P}\Paren{|\langle \bar G_\alpha^\cout, \bar G_j^\cout \rangle| \geq t} \leq 2\exp\Paren{-\Omega(t^2/(n+t B^4))}\,.\]
Taking $t=O(\sqrt{n \log n} + B^4 \log n)$ large enough, we have that $|\langle \bar G_\alpha^\cout, \bar G_j^\cout \rangle| \leq t$ with probability at least $1-1/n^{10}$.
We choose $k$ such that $t \leq k p_+^4 / 2 $, so we have $\langle \bar G_\alpha, \bar G_j\rangle \geq k p_+^4 / 2$.
Finally, an application of the union bound shows that this holds simultaneously for all choices of $\alpha$ and $j$ with probability at least $1-o(1)$.
\end{proof}

Next, we want to prove that each $j \not\in S^*$ satisfies $\langle G_\alpha, G_j\rangle \geq kp_+^4/2$ for few sets $\alpha \subset S^*$ with $|\alpha|=3$. A similar application of Bernstein and union bound can be used to control the maximum of the quantity $\langle \bar G_{\alpha}^\cin, \bar G_j^{\cin}\rangle$ over all $\alpha \subset S^*$ with $|\alpha|=3$ and $j \not\in S^*$.
\begin{lemma}
\label{lem:bad-cin-general}
With high probability, for all $\alpha \subset S^*$ with $|\alpha|=3$ and all $j \not\in S^*$ we have that
\begin{equation*}
|\inprod{\bar G_{\alpha}^\cin, \bar G_j^{\cin}}| \leq O(\sqrt{k \log n} + B^4 \log n)\,.
\end{equation*}
\end{lemma}
\begin{proof}
The element-wise product of $\bar G_{\alpha}^{\cin}$ and $G_j^{\cin}$ has independent mean-zero and variance-one entries with maximum absolute value $B^4$. 
The bound follows by Bernstein's inequality and a union bound as in \Cref{lem:all-good-gen}.
\end{proof}

In order to ensure that, as in the $p=1/2$ case, for each element $j \not\in S^*$ we can bound the number of sets $\alpha \subset S^*$ that lead to large $\inprod{ \bar G_{\alpha}^\cout, \bar G_{j}^\cout}$, we want to appeal again to \Cref{clm:not-many-correlated-vectors}. 
To this end, we define the matrix $H \in \mathbb{R}^{(n-k) \times \Theta(k^3)}$ with columns $G_\alpha^\cout$ for all $\alpha \subset S^*$ with $|\alpha|=3$.

As it turns out, this matrix $H$ also satisfies the conditions of~\Cref{lem:tensoring-gives-RIP}, with a bound $B$ on the size of the entries that generate the matrix. Then, we can conclude:
\begin{lemma}
\label{cor:right-degree}
Let $k \geq O(p_{+}^{-4} B^{4} \sqrt{n} \log^2 n)$. With high probability, for any $j\not\in S^*$, there exist at most $\Oh(p_{+}^{-8} B^2 n^2/k^2)$ sets $\alpha \subset S^*$ with $|\alpha| = 3$ such that $|\inprod{G_{\alpha}, G_j}| \geq k p_{+}^4 / 2$.
\end{lemma}
\begin{proof}
Let $H \in \mathbb{R}^{(n-k) \times \Theta(k^3)}$ be the matrix with columns $G_\alpha^\cout$ for all $\alpha \subset S^*$ with $|\alpha|=3$.
By~\Cref{lem:tensoring-gives-RIP}, the matrix $H/\sqrt{n-k}$ satisfies $(r,O(1))$-RIP for $r = \Omega((n-k)/(B^6\log^4(n)))$.
Then we apply \Cref{clm:not-many-correlated-vectors} for $H/\sqrt{n-k}$ and $\tau = kp_+^4/(3B(n-k))$. Then no unit vector has correlation at least ${\tau = kp_+^4/(3B(n-k))}$ with more than $O(1 / \tau^2) \leq O(p_+^{-8}B^{2}n^2/k^2)$ columns of $H/\sqrt{n-k}$.
Because the norm of each $G_j^\cout$ is bounded by $B\sqrt{n-k}$, we get that no vector $G_j^\cout$ has correlation at least $kp_+^4/3$ with more than $O(p_+^{-8}B^2 n^2/k^2)$ columns of $H$.
We need to ensure that the conditions of \Cref{clm:not-many-correlated-vectors} are satisfied, namely that $r \gtrsim p_+^{-8}B^2 n^2/k^2$. This holds when $k \geq O(p_+^{-4}B^4\sqrt{n} \log^2 n)$.

We can now finish the argument:
If a vector $G_j^\cout$ has correlation at most $kp_+^4/3$ with $G_\alpha^\cout$, then also $\inprod{G_\alpha, G_j} \leq kp_+^4/3 + O(\sqrt{k \log n} + B^4 \log n) < kp_+^4/2$. 
Hence each vector $G_j$ with $j \not\in S^*$ has correlation greater than or equal to $k p_+^4 /2$ with at most $O(p_+^{-8}B^2 n^2/k^2)$ vectors $G_\alpha$.
\end{proof}

We are now ready to finish the proof of the theorem.
\begin{proof}[Proof of \Cref{thm:pgeneral-main}]
We use an identical algorithm as in the $p=1/2$ case. For $\Oh((n/k)^3)$ turns, we take a random triple of vertices $\alpha \subset [n]$ with $|\alpha|=3$ and construct a set $S_{\alpha} := \{ j : \inprod{\bar G_{\alpha}, \bar G_j} \geq k p_+^4/2\}$. Then we construct a refined set $S_{\alpha}' \subseteq S_{\alpha}$ by removing all vertices from $S_{\alpha}$ that are connected to less than $k - 1$ of the vertices in $S_{\alpha}$, and then add $S_{\alpha}'$ to the list if it forms a $k$-clique. Finally, we apply pruning (see \cref{sec:pruning}) to reduce the list size to $(1+o(1))n/k$.

To conclude that this algorithm works, we want to argue that conditioned on $\alpha \subset S^*$, the procedure succeeds in recovering $S^*$ with probability at least $1/2$.

Indeed, by \Cref{lem:all-good-gen}, once $\alpha \subset S^*$ we have with high probability $S^* \subseteq S_\alpha$. 
Moreover, for any $j \not\in S^*$, by \Cref{cor:right-degree}, $\bar G_j$ has large inner product with at most $T = \Oh(p_{+}^{-8} B^2 n^2/k^2)$ vectors $\bar{G}_\alpha$ over $\alpha \subset S^*$ with $|\alpha|=3$. 
There are $\Omega(k^3)$ such sets $\alpha$, so the probability that $\bar{G}_j$ has large inner product with a $\bar{G}_\alpha$ chosen at random is at most $\Oh(T/k^3)$. 
Then, the expected number of $\bar{G}_j$ that have large inner product with a vector $\bar{G}_{\alpha}$ for a random $\alpha \subset S^*$ with $|\alpha|=3$ is $\Oh(nT/k^3)$. 
Hence, by Markov's inequality, with probability $0.99$ over the choice of random $\alpha \subset S^*$ we have that $|S_\alpha \setminus S^*| \leq \Oh(n T/k^3) \leq \Oh(p_{+}^{-8} B^2 n^3/k^5) \leq o(k(1-p))$.

Once this is the case, note that removing vertices of degree smaller than $k-1$ in the induced graph given by $S_\alpha$ clearly does not affect vertices in $S^*$. On the other hand, for any $j \not \in S^*$, denoting by $N(j)$ the neighborhood of the vertex $j$ in $G$, we have
\begin{align*}
|N(j) \cap S_{\alpha}|
&\leq |N(j) \cap S^*| + |N(j) \cap (S_\alpha \setminus S^*)|\\
&\leq kp + \Oh\Paren{\sqrt{k p (1-p) \log n}} + o(k(1-p))\\
&< k-1\,.
\end{align*}
Then the pruning algorithm in \Cref{sec:pruning} keeps $S^*$ in the final list.

The time complexity is the same as in the $p=1/2$ case, using matrix multiplication.
\end{proof}

%% file: discussion.tex
\section{Discussion}%
\label{sec:discussion}

We include some commentary and open questions that naturally arise given our new approach. 

\textbf{1) Handling a monotone adversary.} In this work, we gave a simple algorithm that solves the semirandom planted clique recovery problem at nearly the right threshold of $k = \tilde{O}(\sqrt{n})$. Unlike the prior works~\cite{MR4617517-Buhai23,MMT20} (we note that the algorithm in~\cite{DBLP:conf/stoc/CharikarSV17,steinhardt2017does} does not tolerate monotone deletions), our algorithm does not handle a monotone adversary that can delete an arbitrary subset of edges in the cut defined by the planted clique. 

In prior works, approaches that handle a monotone adversary~\cite{DBLP:journals/rsa/FeigeK00,DBLP:conf/stoc/MoitraPW16} have been naturally based on semidefinite programming. In fact, the analyses of such semidefinite programs can be naturally interpreted as yielding an efficient certificate of some natural property that controls the uniqueness of the relevant solution concept (e.g., clique, community, coloring). In our context, such a certificate of uniqueness naturally corresponds to the biclique numbers for random bipartite graphs studied in~\cite{MR4617517-Buhai23}. Indeed, they provide some evidence of non-existence of efficient certificates for this property by means of lower bounds in the low-degree polynomial model. Our key conceptual contribution is finding an algorithmic approach that circumvents the need for such certificates.  

We are thus left with an intriguing question: is there an algorithm for semirandom planted clique that succeeds at $k = \tilde{O}(\sqrt{n})$ \emph{and} tolerates a monotone adversary? Can such an algorithm yield an efficient (e.g., low-degree sum-of-squares) certificate for biclique bounds on random bipartite graphs and circumvent the low-degree polynomial lower bounds? 

On the flip side, could there be a computational--statistical complexity gap for semirandom planted clique that arises entirely due to a monotone adversary? This would run counter to the central thesis of~\cite{steinhardt2017does} that suggests that robustness (i.e., success in the semirandom model) might obliterate the computational--statistical complexity gap for the fully random planted clique problem. We \emph{do not} know of any problem for which a monotone adversary ``creates'' such a gap. We note that for community detection in stochastic block models, the work of Moitra, Perry and Wein~\cite{DBLP:conf/stoc/MoitraPW16} showed that monotone deletions shift the \textit{information-theoretic} threshold.

\textbf{2) Can our algorithm be captured by semidefinite programming?}
In light of the discussion above, it is also natural to ask: can our algorithm be ``captured'' (i.e., suggest an analysis with the same guarantees) by a natural semidefinite programming relaxation for the semirandom planted clique problem? We do not know how to do this so far. As discussed above, analyses of semidefinite programming relaxations usually yield certificates for an underlying property (biclique number bounds for random bipartite graphs in our setting) that implies uniqueness of the relevant solution concept. We are aware of two exceptions to this ``rule'' where algorithms based on SDPs succeed despite the potential impossibility of the associated certification problem: 1) the SDP relaxations for low-rank matrix sensing~\cite{doi:10.1137/070697835} without a certificate of RIP, 2) robust mean and covariance estimation for Gaussian distributions~\cite{kothari2022polynomial} without a certificate of ``resilience''. Is there a similar approach that manages to recover our guarantees for semirandom planted clique without the need for improved certificates for biclique numbers in random bipartite graphs?

\textbf{3) Achieving $k=O(\sqrt{n \log n})$?} Our algorithm currently needs $k \geq O(\sqrt{n} \log^2 n)$, which is $\log^{1.5}(n)$ off of the likely ``right'' bound (i.e., matching the guarantee of the best-known brute-force algorithm). This $\log^{1.5}(n)$-factor loss arises from the loss in the sparsity parameter in~\Cref{lem:tensoring-gives-RIP}. We believe that for our matrix $H$ this loss can likely be removed and consequently our algorithm could succeed at $k \geq O(\sqrt{n\log n})$. We note that the general result (\Cref{thm:rip}) we rely on cannot be improved to provide such a guarantee: there exists an isotropic distribution over bounded vectors from which it is necessary to sample at least $\Omega(\sqrt{n} \log^2 n)$ rows in order to obtain RIP at sparsity $\sqrt{n}$ \cite{Blasiok2023Improved}. However, our application only needs the weaker \emph{sparse operator norm} bounds to which the lower bounds in~\cite{Blasiok2023Improved} do not apply. While this gap has not been consequential in known results, it does leave open the possibility of a stronger, general variant of \cref{thm:rip} for sparse operator norm bounds.

%% file: appendix.tex
\section{Restricted isometry property}
\label{sec:rip}
One of the main tools used in the analysis of our algorithm is a strong general theorem guaranteeing RIP for any matrix with independent rows drawn from any isotropic-distribution vectors with bounded entries.

\Cref{thm:rip}, although known, is usually stated in terms of so-called \emph{bounded orthonormal systems}, a formulation that makes its direct applicability in situations like~\Cref{lem:tensoring-gives-RIP} less transparent. 
For completeness, we introduce the notion of bounded orthonormal system, the formulation of~\Cref{thm:rip} as it is stated in \cite{MR2731597-Rauhut10}, and we discuss why this is equivalent to the statement we are using in this paper.

\begin{definition}
\label{def:bos}
A bounded orthonormal system is given by a region $\mathcal{D} \subset \R^k$, together with a probability measure $\mu$ on $\mathcal{D}$ and $N$ functions $\phi_1,\phi_2, \ldots \phi_d : \mathcal{D} \to \C$ satisfying the following properties:
\begin{itemize}
\item For all $i\not=j$ in $[N]$,
$\E_{t \sim \mu} \phi_i(t) \overline{\phi_j(t)} = 0$,
\item For all $i\in [N]$,
$\E_{t \sim \mu} \phi_i(t) = 0, \E_{t\sim \mu} |\phi_i(t)|^2 = 1$,
\item For all $i\in [N]$, $|\phi_i(t)| \leq K$ almost surely.
\end{itemize}
\end{definition}
Then Theorem 12.31 in \cite{FR13} (originally proved in~\cite{MR2417886-Rudelson08}) states:
\begin{theorem}[Theorem 12.31 in \cite{FR13}]
\label{thm:rip-bos}
If $\phi_1, \ldots, \phi_N$ together with a distribution $\mu$ is a bounded orthonormal system, then a matrix $H \in \C^{m \times N}$ obtained by sampling $t_1, t_2 \ldots t_m$ independently at random according to $\mu$ and setting $H_{i,j} = \phi_j(t_i)/\sqrt{m}$ satisfies $(r,\delta)$-RIP with probability $1-N^{-\log^3(r)}$ whenever
\begin{equation*}
m \gtrsim K^2 r \log^3 r \log(N) / \delta^2\,.
\end{equation*}
\end{theorem}
To see how this theorem implies \Cref{thm:rip}, let us consider an arbitrary distribution $\mu$ over $\R^n$ satisfying $\E_{\mathbf{X} \sim \mu} \mathbf{X} \mathbf{X}^\top = I$. We can take $\phi_1, \ldots \phi_N : \R^n \to \R$ such that $\phi_j$ is a projection on the $j$-th coordinate. Then $\phi_1, \ldots \phi_N$ together with the distribution $\mu$ forms a bounded orthonormal system according to \Cref{def:bos}, and the matrix $H_{i,j} = \phi_j(t_i)$ is exactly obtained by drawing $m$ independent random rows according to the law of $\mathbf{X}$.

The other direction of the equivalence (in the case of real-valued bounded orthonormal system) is just as immediate: if $\phi_1, \ldots, \phi_N$ together with $\mu$ form a bounded orthonormal system, then a random vector $\mathbf{X}$ obtained by drawing $t \sim \mu$ and setting $\mathbf{X} = (\phi_1(t), \ldots \phi_N(t))$ is a bounded vector in isotropic position.

\section{Pruning a list of $k$-cliques} \label{sec:pruning}

We start by recalling the following two lemmas from \cite{MR4617517-Buhai23}:

\begin{lemma}[Lemma 5.7 in \cite{MR4617517-Buhai23}]
\label{lem:clique-intersection}
Let $G \sim \SRPC(n,k,p)$. Let $S^*$ be the planted clique in $G$.
Then, with probability at least $1-\frac{k}{n^2}$, any other clique $S$ of size at least $k$ satisfies $|S\cap S^*| \leq 3 \frac{\log n}{\log 1/p}$.
\end{lemma}

\begin{lemma}[Lemma 5.8 in \cite{MR4617517-Buhai23}]
\label{lem:sets-pinex}
Let $S_1, ..., S_m \subseteq [n]$ with $|S_i| = k$ and $|S_i \cap S_j| \leq \Delta$.
Then, if $k \geq \sqrt{2n\Delta}$, we have $m \leq \frac{n}{k} \Paren{1+\frac{2n\Delta}{k^2}}$.
\end{lemma}

Now we prove the pruning result:

\begin{lemma}[Pruning, implicit in \cite{MR4617517-Buhai23}]
\label{lem:pruning}
Let $G$ be a graph on $n$ vertices generated according to $\SRPC(n,k,p)$ where $k \geq C \sqrt{n \log n \cdot \max(1, p/(1-p))}$ for some large enough absolute constant $C$. Let $L$ be a list of size $m$ of $k$-cliques of $G$. 

Then taking $L'$ to be the list of all cliques from $L$ which have intersection  at most $3 \log n / \log p^{-1}$ with all other cliques in $L$, we have with high probability
\begin{itemize}
\item If $S^* \in L$, then also $S^* \in L'$, and
\item $|L'| \leq (1 + O(1/C^2)) n/k$.
\end{itemize}
\end{lemma}
\begin{proof}
First, we have from~\Cref{lem:clique-intersection} that with high probability the planted clique has intersection at most $3 \frac{\log n}{\log 1/p}$ with any other clique in $G$, so we are guaranteed that we do not remove it from the list.

Second, all the $k$-cliques in the final list have intersection at most $3 \frac{\log n}{\log 1/p}$, so by~\Cref{lem:sets-pinex} the size of the list is at most $\frac{n}{k}\Paren{1+\frac{6n \frac{\log n}{\log 1/p}}{k^2}}$.
It remains to show that $\frac{6n \frac{\log n}{\log 1/p}}{k^2} = O(1/C^2)$ for our choice of $k$.
We use that $\log 1/p \geq (1-p)/4$ for $p\leq 1$ to get  
\begin{align*}
\frac{6n \frac{\log n}{\log 1/p}}{k^2}
&\leq \frac{24n \frac{\log n}{1-p}}{C^2 n \log n \cdot \max(1, p/(1-p))}\\
&= \frac{24}{C^2 (1-p) \max(1, p/(1-p))}\\
&= \frac{24}{C^2 \max(1-p, p)}
\leq \frac{48}{C^2}\,,
\end{align*}
which gives the desired result.
\end{proof}

\begin{lemma}[Fast pruning]
\label{lem:pruning-fast}
The pruning procedure described in~\Cref{lem:pruning} can be implemented in time $O(\lceil m/n \rceil n^{\omega}\frac{\log m}{\log k})$ by using $O(\lceil m/n\rceil \cdot \frac{\log m}{\log k})$ calls to the $n\times n$ matrix multiplication oracle.
\end{lemma}
\begin{proof}
We implement the algorithm using matrix multiplication as follows:
\begin{itemize}
    \item While $m > (n/k)(1+o(1))$:
    \begin{itemize}
        \item Split the list of $m$ $k$-cliques arbitrarily into $\ceil{m/n}$ lists of at most $n$ $k$-cliques,
        \item For each of the $\ceil{m/n}$ lists of at most $n$ $k$-cliques:
        \begin{itemize}
            \item Compute $U \in \{0,1\}^{n \times n}$ such that $U_{i,j} = 1$ if and only if the $i$-th clique in the list contains vertex $j$,
            \item Compute $M = UU^\top$ in time $O(n^{\omega})$,
            \item For each clique $i$ in the list, iterate over all cliques $j > i$ in the list, and if $M_{i,j} > 3 \frac{\log n}{\log 1/p}$, mark clique $j$ as removed,
        \end{itemize}
        \item Construct a list of all cliques that are not removed in any of the $\ceil{m/n}$ lists.
        By~\Cref{lem:pruning}, this list has size at most $\ceil{m/n} \cdot (n/k)(1+o(1)) \leq (m/k)(1+o(1))$.
        Update $m$ to be the size of this new list.
    \end{itemize}
\end{itemize}

The while loop performs at most $\log_k m = \frac{\log m}{\log k}$ iterations, and the inner loop has time complexity $O(\lceil m/n \rceil n^{\omega}+mn) \leq O(\lceil m/n\rceil n^{\omega})$.
Therefore the total time complexity is $O(\lceil m/n\rceil n^{\omega}\frac{\log m}{\log k})$.

\end{proof}

\section{Adversarial correlation with independent Rademacher vectors} \label{sec:in-prod-corr}

\begin{lemma}
\label{lem:adv-cor}
Let $v_1, ..., v_m \in \{\pm 1\}^n$ be independent Rademacher vectors. 
Then with high probability there exists some vector $w \in \{\pm 1\}^n$ such that $\langle v_i, w\rangle \geq \Omega(n/\sqrt{m})$ for all $i \in [m]$ as long as $n \gg m \log m$.

In other words, for any $k = \Omega(n/\sqrt{m})$ small enough (i.e., $m = \Omega(n^2/k^2)$ small enough), with high probability there exists some vector $w \in \{\pm 1\}^n$ such that $\langle v_i, w\rangle \geq k$ for all $i \in [m]$ as long as $n \gg m \log m$.
\end{lemma}
\begin{proof}
We consider without loss of generality the case when $m$ is odd.
Let $w_j = \mathrm{Maj}(v_{1, j}, \ldots, v_{m, j})$ for all $j \in [n]$, i.e., the majority element.
Note that $\langle v_i, w\rangle = 2 \sum_{j=1}^n \mathbf{1}_{w_j = v_{i,j}} - n$, where $\mathbf{1}_{w_j = v_{i,j}}$ are independent $\mathrm{Ber}(p)$ random variables, with $p$ the probability that $v_{i,j} = \mathrm{Maj}(v_{1,j}, \ldots, v_{m, j})$.
Note that if $v_{1, j}, \ldots, v_{i-1,j}, v_{i+1,j}, \ldots, v_{m, j}$ already forms a majority of $\geq \frac{m+1}{2}$ elements, then $v_{i,j}$ cannot influence the majority and it will be in the majority with probability $1/2$.
Otherwise, $v_{1, j}, \ldots, v_{i-1,j}, v_{i+1,j}, \ldots, v_{m, j}$ must have an equal number of $1$ and $-1$ elements, and $v_{i, j}$ will always be in the majority. 
The latter case happens with probability $2\binom{m-1}{(m-1)/2}/2^{m-1}$, which grows asymptotically as $\frac{2}{\sqrt{\pi m / 2}} (1\pm o(1))$ (e.g., see the growth of the central binomial coefficient in \cite{luke1969special}, page 35).
Therefore $p = 1/2+\Omega(1/\sqrt{m})$.
From this, by Hoeffding's inequality, we get that ${\mathbb{P}(\langle v_i, w\rangle - \Omega(n/\sqrt{m}) \geq t)} \leq \exp(-\Omega(t^2/n))$, so with high probability ${\langle w, v_i\rangle \geq \Omega(n/\sqrt{m}) - O(\sqrt{n \log m})}$ for all $v_1, \ldots, v_m$.
The first term dominates as long as ${n \gg m \log m}$.
\end{proof}

%% file: main.bbl
\newcommand{\etalchar}[1]{$^{#1}$}
\begin{thebibliography}{DKWB21}

\bibitem[AKS98]{MR1662795-Alon98}
Noga Alon, Michael Krivelevich, and Benny Sudakov.
\newblock Finding a large hidden clique in a random graph.
\newblock In {\em Proceedings of the {E}ighth {I}nternational {C}onference
  ``{R}andom {S}tructures and {A}lgorithms'' ({P}oznan, 1997)}, volume~13,
  pages 457--466, 1998.

\bibitem[BB20]{brennan2020reducibility}
Matthew Brennan and Guy Bresler.
\newblock Reducibility and statistical-computational gaps from secret leakage.
\newblock In {\em Conference on Learning Theory}, pages 648--847. PMLR, 2020.

\bibitem[BBH18]{brennan2018reducibility}
Matthew Brennan, Guy Bresler, and Wasim Huleihel.
\newblock Reducibility and computational lower bounds for problems with planted
  sparse structure.
\newblock In {\em Conference On Learning Theory}, pages 48--166. PMLR, 2018.

\bibitem[BDMS13]{DBLP:journals/tit/BandeiraDMS13}
Afonso~S. Bandeira, Edgar Dobriban, Dustin~G. Mixon, and William~F. Sawin.
\newblock Certifying the restricted isometry property is hard.
\newblock {\em {IEEE} Trans. Information Theory}, 59(6):3448--3450, 2013.

\bibitem[BHK{\etalchar{+}}16]{BHK+16}
Boaz Barak, Samuel~B. Hopkins, Jonathan Kelner, Pravesh~K. Kothari, Ankur
  Moitra, and Aaron Potechin.
\newblock A {N}early {T}ight {S}um-of-{S}quares {L}ower {B}ound for the
  {P}lanted {C}lique {P}roblem.
\newblock In {\em Proceedings of the 57th Annual IEEE Symposium on Foundations
  of Computer Science}, 2016.

\bibitem[BKS23]{MR4617517-Buhai23}
Rares-Darius Buhai, Pravesh~K. Kothari, and David Steurer.
\newblock Algorithms approaching the threshold for semi-random planted clique.
\newblock In {\em S{TOC}'23---{P}roceedings of the 55th {A}nnual {ACM}
  {S}ymposium on {T}heory of {C}omputing}, pages 1918--1926. ACM, New York,
  [2023] \copyright 2023.

\bibitem[BLL{\etalchar{+}}23]{Blasiok2023Improved}
Jaroslaw Blasiok, Kyle Luh, Patrick Lopatto, Jake Marcinek, and Shravas Rao.
\newblock An improved lower bound for sparse reconstruction from subsampled
  walsh matrices.
\newblock {\em Discrete Analysis}, 5 2023.

\bibitem[BR13a]{DBLP:conf/colt/BerthetR13}
Quentin Berthet and Philippe Rigollet.
\newblock Complexity theoretic lower bounds for sparse principal component
  detection.
\newblock In {\em {COLT}}, volume~30 of {\em {JMLR} Workshop and Conference
  Proceedings}, pages 1046--1066. JMLR.org, 2013.

\bibitem[BR13b]{berthet2013computational}
Quentin Berthet and Philippe Rigollet.
\newblock Computational lower bounds for sparse pca.
\newblock {\em arXiv preprint arXiv:1304.0828}, 2013.

\bibitem[BS95]{BLUM1995204}
A.~Blum and J.~Spencer.
\newblock Coloring random and semi-random k-colorable graphs.
\newblock {\em Journal of Algorithms}, 19(2):204 -- 234, 1995.

\bibitem[CSV17]{DBLP:conf/stoc/CharikarSV17}
Moses Charikar, Jacob Steinhardt, and Gregory Valiant.
\newblock Learning from untrusted data.
\newblock In {\em {STOC}}, pages 47--60. {ACM}, 2017.

\bibitem[DKWB21]{MR4345126-Ding21}
Yunzi Ding, Dmitriy Kunisky, Alexander~S. Wein, and Afonso~S. Bandeira.
\newblock The average-case time complexity of certifying the restricted
  isometry property.
\newblock {\em IEEE Trans. Inform. Theory}, 67(11):7355--7361, 2021.

\bibitem[Fei19]{bwca-semirandom}
Uriel Feige.
\newblock Introduction to semirandom models.
\newblock In Tim Roughgarden, editor, {\em Beyond Worst-case Analysis of
  Algorithms}, chapter~10, pages 266--290. Oxford, 2019.

\bibitem[FGR{\etalchar{+}}17]{MR3664576-Feldman17}
Vitaly Feldman, Elena Grigorescu, Lev Reyzin, Santosh~S. Vempala, and Ying
  Xiao.
\newblock Statistical algorithms and a lower bound for detecting planted
  cliques.
\newblock {\em J. ACM}, 64(2):Art. 8, 37, 2017.

\bibitem[FK00]{DBLP:journals/rsa/FeigeK00}
Uriel Feige and Robert Krauthgamer.
\newblock Finding and certifying a large hidden clique in a semirandom graph.
\newblock {\em Random Struct. Algorithms}, 16(2):195--208, 2000.

\bibitem[FK01]{FK01}
Uriel Feige and Joe Kilian.
\newblock Heuristics for semirandom graph problems.
\newblock {\em Journal of Computer and System Sciences}, 63(4):639 -- 671,
  2001.

\bibitem[FR13]{FR13}
Simon Foucart and Holger Rauhut.
\newblock {\em A Mathematical Introduction to Compressive Sensing}.
\newblock Birkh\"{a}user Basel, 2013.

\bibitem[H{\aa}s99]{MR1687331-Haastad99}
Johan H{\aa}stad.
\newblock Clique is hard to approximate within {$n^{1-\epsilon}$}.
\newblock {\em Acta Math.}, 182(1):105--142, 1999.

\bibitem[HKP{\etalchar{+}}17]{HKP+17}
Samuel~B Hopkins, Pravesh~K Kothari, Aaron Potechin, Prasad Raghavendra, Tselil
  Schramm, and David Steurer.
\newblock The power of sum-of-squares for detecting hidden structures.
\newblock In {\em 2017 IEEE 58th Annual Symposium on Foundations of Computer
  Science (FOCS)}, pages 720--731. IEEE, 2017.

\bibitem[Jer92]{DBLP:journals/rsa/Jerrum92}
Mark Jerrum.
\newblock Large cliques elude the metropolis process.
\newblock {\em Random Struct. Algorithms}, 3(4):347--360, 1992.

\bibitem[KMZ22]{kothari2022polynomial}
Pravesh~K Kothari, Peter Manohar, and Brian~Hu Zhang.
\newblock Polynomial-time sum-of-squares can robustly estimate mean and
  covariance of gaussians optimally.
\newblock In {\em International Conference on Algorithmic Learning Theory},
  pages 638--667. PMLR, 2022.

\bibitem[Kuc95]{DBLP:journals/dam/Kucera95}
Ludek Kucera.
\newblock Expected complexity of graph partitioning problems.
\newblock {\em Discrete Applied Mathematics}, 57(2-3):193--212, 1995.

\bibitem[KZ14a]{DBLP:journals/tit/KoiranZ14}
Pascal Koiran and Anastasios Zouzias.
\newblock Hidden cliques and the certification of the restricted isometry
  property.
\newblock {\em {IEEE} Trans. Information Theory}, 60(8):4999--5006, 2014.

\bibitem[KZ14b]{MR3245368-Koiran14}
Pascal Koiran and Anastasios Zouzias.
\newblock Hidden cliques and the certification of the restricted isometry
  property.
\newblock {\em IEEE Trans. Inform. Theory}, 60(8):4999--5006, 2014.

\bibitem[Luk69]{luke1969special}
Yudell~L. Luke.
\newblock Chapter {II}. {T}he {G}amma {F}unction and {R}elated {F}unctions.
\newblock In {\em The Special Functions and Their Approximations}, volume~53 of
  {\em Mathematics in Science and Engineering}, pages 8--37. Elsevier, 1969.

\bibitem[MMT20]{MMT20}
Theo McKenzie, Hermish Mehta, and Luca Trevisan.
\newblock A new algorithm for the robust semi-random independent set problem.
\newblock In Shuchi Chawla, editor, {\em Proceedings of the 2020 {ACM-SIAM}
  Symposium on Discrete Algorithms}, pages 738--746, 2020.

\bibitem[MPW16]{DBLP:conf/stoc/MoitraPW16}
Ankur Moitra, William Perry, and Alexander~S. Wein.
\newblock How robust are reconstruction thresholds for community detection?
\newblock In {\em {STOC}}, pages 828--841. {ACM}, 2016.

\bibitem[Rau10]{MR2731597-Rauhut10}
Holger Rauhut.
\newblock Compressive sensing and structured random matrices.
\newblock In {\em Theoretical foundations and numerical methods for sparse
  recovery}, volume~9 of {\em Radon Ser. Comput. Appl. Math.}, pages 1--92.
  Walter de Gruyter, Berlin, 2010.

\bibitem[RFP10]{doi:10.1137/070697835}
Benjamin Recht, Maryam Fazel, and Pablo~A. Parrilo.
\newblock Guaranteed minimum-rank solutions of linear matrix equations via
  nuclear norm minimization.
\newblock {\em SIAM Review}, 52(3):471--501, 2010.

\bibitem[RV08]{MR2417886-Rudelson08}
Mark Rudelson and Roman Vershynin.
\newblock On sparse reconstruction from {F}ourier and {G}aussian measurements.
\newblock {\em Comm. Pure Appl. Math.}, 61(8):1025--1045, 2008.

\bibitem[Ste17]{steinhardt2017does}
Jacob Steinhardt.
\newblock Does robustness imply tractability? a lower bound for planted clique
  in the semi-random model.
\newblock {\em arXiv preprint arXiv:1704.05120}, 2017.

\bibitem[Zuc07]{MR2403018-Zuckerman07}
David Zuckerman.
\newblock Linear degree extractors and the inapproximability of max clique and
  chromatic number.
\newblock {\em Theory Comput.}, 3:103--128, 2007.

\end{thebibliography}
